\newcommand{\rd}{{\rm d}}

\newcommand{\vecx}{{\mathbf x}}
\newcommand{\vecX}{{\mathbf X}}

\newcommand{\vecw}{{\mathbf w}}

\documentclass[11pt]{article}
\usepackage{graphicx}
\usepackage{placeins}
\usepackage{longtable}
\usepackage{float}
\usepackage{longtable,tabu}
\usepackage{epstopdf}
\usepackage{relsize,exscale}
\usepackage{amssymb}
\usepackage{bm}
\usepackage{amsmath}
\usepackage{cancel}
\usepackage[margin=1.0in]{geometry}
\usepackage{cite}
\usepackage[shortlabels]{enumitem}
\usepackage{scalerel}
\usepackage{color}
\usepackage{enumitem}
\usepackage[titletoc]{appendix}
\usepackage[font=small,labelfont=bf]{caption}
\usepackage{subcaption}
\allowdisplaybreaks
\usepackage{makecell}
\usepackage{amsthm}
\usepackage{mathtools}
\usepackage{array}
\usepackage{longtable}
\newtheorem{theorem}{Theorem}

\newtheorem{lemma}{Lemma}

\newtheorem{corollary}{Corollary}

\usepackage{chngcntr}
\usepackage{bm}
\counterwithout{theorem}{section}
\counterwithout{definition}{section}
\counterwithout{lemma}{section}
\counterwithout{remark}{section}
\counterwithout{assumption}{section}
\counterwithout{proposition}{section}
\counterwithout{corollary}{section}
\counterwithout{claim}{section}
\counterwithout{conjecture}{section}
\begin{document}
\title{Numerical evaluation of Gaussian mixture entropy}
\author{Basheer Joudeh and Boris \v{S}kori\'{c}}
\date{}
\maketitle
\begin{abstract}
We develop an approximation method for the differential entropy $h(\vecX)$ of a $q$-component Gaussian mixture in $\mathbb{R}^n$. We provide two examples of approximations using our method denoted by $\bar{h}^{\mathrm{Taylor}}_{C,m}(\vecX)$ and $\bar{h}^{\mathrm{Polyfit}}_{C}(\vecX)$. We show that $\bar{h}^{\mathrm{Taylor}}_{C,m}(\vecX)$ provides an easy to compute lower bound to $h(\vecX)$, while $\bar{h}^{\mathrm{Polyfit}}_{C}(\vecX)$ provides an accurate and efficient approximation to $h(\vecX)$. $\bar{h}^{\mathrm{Polyfit}}_{C}(\vecX)$ is more accurate than known bounds, and conjectured to be much more resilient than the approximation of \cite{approx} in high dimensions.
\end{abstract}
\section{Introduction}
\subsection{Differential entropy of Gaussian mixtures}
A Gaussian mixture is a probability density function on ${\mathbb R}^n$ of the form
$f(\vecx)=\sum_{j=1}^q p_j {\cal N}_{\vecw_j,K_j}(\vecx)$,
where $\vecx\in{\mathbb R}^n$,
the $p_j>0$ are weights satisfying $\sum_{j=1}^q p_j=1$,
and ${\cal N}_{\vecw,K}$ stands for a Gaussian distribution with mean $\vecw$ and covariance matrix~$K$.
In words, $f$ is a mixture of $q$ Gaussian pdfs, with arbitrary weights, allowing
each Gaussian to be of different shape and displacement.
A mixture like this occurs e.g. when a stochastic process, with probability mass function $p_1,\ldots,p_q$,
determines which distribution holds for~$\vecx$.
Such situations occur in many scientific areas.
Furthermore, Gaussian mixtures are often employed as function approximators, by virtue of being
smooth and localized.
They have been used in a wide variety of studies, e.g.\;on
diffusion models in physics and machine learning \cite{DIFFREF2,DIFFREF1,DIFFREF4,DIFFREF3},
non-adiabatic thermodynamics \cite{THERMREF1},
wireless communication \cite{WIREF1},
wireless authentication \cite{CYBREF2},
Byzantine attacks \cite{CYBREF1},
gene expression \cite{BIOREF1,BIOREF2},
dark matter kinematics \cite{ASTROREF1},
and quasar spectra \cite{ASTROREF2}.

The differential entropy $h(X)$ of a continuous random variable $X\in{\cal X}$, with probability density function $f_X$, is defined as
$h(X)=-\int_{\cal X}{\rm d}x\; f_X(x)\ln f_X(x)$.
It represents the continuum limit of the (discrete) Shannon entropy for the
probability mass function $f_X(x_i)\triangle x$, where the volume $\triangle x$ is sent to zero
and the infinite contribution $\ln \frac1{\triangle x}$ is subtracted.

Analytically computing or estimating the differential entropy of a Gaussian mixture
is a notoriously difficult problem.
Even numerical evaluation can be problematic in high dimensions.
The special case of a single Gaussian is simple and yields
$h(X)=\frac 12\ln\det( 2\pi e  K)$.

\subsection{Related work}

Various methods have been proposed to approximate the differential entropy of a Gaussian mixture.

A loose upper bound can be obtained from the fact that
a Gaussian distribution maximizes entropy, given the first and second moment.
This yields \cite{bound}
$h(X) \leq \frac 12\ln\det(2\pi e\Sigma)$,
with
$\Sigma=\sum_{i=1}^q p_i (w_i w_i^{\rm T}+K_i)-\left(\sum_{i=1}^q p_i w_i\right)(\sum_{j=1}^q p_j w_j)^{\rm T}$.

In \cite{q2} a numerical approximation was given for the case $q=2,n=1,w_2=-w_1$.

A sequence of upper bounds for arbitrary $q$ was obtained in \cite{bound},
but only for $n=1$. Furthermore, the results are not in closed form, and the bounds in the sequence do not get progressively tighter.

Another method is to replace the density $f$ inside the logarithm by
a single Gaussian $\bar f$ which has mean and covariance equal to the mixture.
This leads to an exact expression containing the relative entropy,
$h(X)=\frac 12\ln \det (2\pi e\Sigma)-D(f||\bar f)$;
one can then find approximations or bounds for the relative entropy, as in \cite{relat,relat2}.
This is done with Monte Carlo integration, which has the drawback of being computationally demanding and
not giving an analytic expression.

In \cite{approx} an approximation was obtained by performing a Taylor expansion of $\ln f(\vecx)$
in the variables $\vecx-\vecw_j$.
To avoid the need for expansion powers above 2,
they introduced the trick of representing wide Gaussians approximately as
the sum of several narrow Gaussians, in the $f$ outside the logarithm.
The Taylor expansion as well as the splitting trick introduce inaccuracies.

As shown in \cite{approx}, one can obtain a basic upper bound
$h(X)\leq  \sum_{j=1}^{q} p_j\ln \frac1{p_j} + \sum_{j=1}^{q} p_j\frac 12\ln\det (2\pi e K_j)$.
This is significantly tighter than the above mentioned bound
$\frac 12\ln [(2\pi e)^n\det \Sigma]$;
it is exact in the case of a single Gaussian,
and it gets arbitrarily close to the true value of $h(X)$ when the overlap between the components of the mixture becomes negligible.
A refinement was also introduced by means of merging parts of the mixture that are clustered together. However, in configurations where the mixture components have significant overlap whilst keeping the mixture non-trivial, this bound becomes inaccurate.

In \cite{JoudehSkoric} a method was proposed to compute $h(X)$ averaged over the means~$\vecw_1,\ldots,\vecw_q$.
This is of limited use when the means are fixed.

\subsection{Contributions}
We introduce a new method for estimating the differential entropy of a Gaussian mixture.
We approximate $f\ln\frac1f$ by a polynomial ${\rm poly}(f)$.
For any positive integer $k$, the power $f^k=(p_1 {\cal N}_1+\ldots+p_q{\cal N}_q)^k$ can be
written as a multinomial sum containing a product of powers of Gaussians.
Hence the integral $\int\!\rd x\; {\rm poly}(f(x))$ can be computed analytically.
This yields a systematic way to approximate and/or lower bound~$h(X)$ by analytic expressions.
\begin{itemize}[leftmargin=4mm,itemsep=0mm]
\item
We consider polynomials that minimize the square error
$\int\!\rd s\; {\rm weight}(s)[s\ln\frac1s-{\rm poly}(s)]^2$
on the range of $f$, with tunable function ${\rm weight}(s)=s^r$.
We observe that negative $r$ gives better results than positive $r$,
which is to be expected since most of the volume in ${\mathbb R}^n$ has $f(x)$ close to zero.
In particular, $r\approx-2$ performs best in the mixture configurations that we have studied,
yielding relative errors in $h(X)$ of less than 1\% already for the degree-3 polynomial fit.
\item
We consider the truncated Taylor series
$-\ln \frac fm\approx \sum_{k=1}^{\cdots} \frac1k(1-\frac fm)^k$
for some constant~$m$.
This too gives rise to an approximation for $f\ln\frac1f$ that is polynomial in $f$.
While not as accurate as the above mentioned fit, it guarantees that each $k$-term has the same sign
if $m\geq \max_{x\in{\mathbb R}^n} f(x) $,
and hence gives rise to a sequence of increasingly accurate analytic lower bounds on~$h(X)$. We observe that the relative error in $h(X)$ is still several percent even at polynomial degree~10.
\end{itemize}

In Section \ref{generalpolsec} we obtain a general recipe to approximate the differential entropy via a polynomial approximation as shown in Corollary \ref{fpowerintegralcorr}. In Section \ref{Taylorsec} we apply Corollary \ref{fpowerintegralcorr} to the Taylor series of the logarithm to obtain a specific approximation for the differential entropy. Although it is not a very accurate approximation, it serves as an easy to compute lower bound. In Section \ref{polyfitsec} we apply Corollary \ref{fpowerintegralcorr} again using a polynomial approximation of $f(s)=-s \ln s$, which gives an accurate approximation to the differential entropy. In Section \ref{resultsec} we show numerical results of our approximations for various configurations, and in Section \ref{discsec} we assess our method and give pointers towards future work.

\section{Polynomial approximation}
\subsection{General polynomial}
\label{generalpolsec}
Consider a random variable $\mathbf{X}\in{\mathbb R}^n$ whose probability density function is a Gaussian mixture with weights $\{p_i\}_{i=1}^{q}$. The Gaussian pdfs have covariance matrices $\{K_i\}_{i=1}^{q}$, and they are centered on points $\vecw_1,\ldots,\vecw_q \in {\mathbb R}^n$, and we write $\hat\vecw=(\vecw_1,\ldots,\vecw_q)$.
\begin{equation}
\label{gaussmix}
    f_{\vecX}(\vecx)
    = \sum_{j=1}^q p_j{\cal N}_{\vecw_j,K_j}(\vecx)
    = (2\pi)^{-\frac n2}  \sum_{j=1}^q p_j (\det K_j)^{-1/2}\exp \left\{- \dfrac{1}{2}(\vecx-\vecw_j)^\mathrm T K_j^{-1}(\vecx-\vecw_j)\right\}.
\end{equation}
We let $g_i(\vecx) \equiv {\cal N}_{\vecw_{i},K_i}(\vecx)$ in the following. The differential entropy is $h(\vecX)=-\mathbb{E}_{\mathbf X}[\ln f_{\mathbf X}]$, which can be written as:
\begin{equation}
\begin{aligned}
h(\vecX)=-\mathbb{E}_{\mathbf X}[\ln f_{\mathbf X}]=-\mathbb{E}_{\mathbf X}[\ln mm^{-1}f_{\mathbf X}]=-\ln m -\mathbb{E}_{\mathbf X}\left[\ln \dfrac{ f_{\mathbf X}}{m}\right],
\end{aligned}
\end{equation}
where $m$ can be chosen to enforce that the argument of the logarithm lies in a certain interval. We would like to approximate the entropy by means of a polynomial approximation in powers of $f_{\mathbf X}$:
\begin{equation}\label{polyfit0}
  -f_{\mathbf{X}}\ln \dfrac{f_{\mathbf{X}}}{m}\approx\sum_{a=1}^{C}{c_a} f_{\mathbf{X}}^a,
\end{equation}
i.e. we consider an order-$C$ polynomial approximation. We set $c_0 = 0$ to avoid diverging integrals. Equation \eqref{polyfit0} is more relaxed than truncating a series of the form $\sum_{a=1}^{\infty}{c_a} f_{\mathbf{X}}^a$, i.e. we allow $\{c_a\}$ to depend on $C$. Using equation \eqref{polyfit0}, the differential entropy $h(\vecX)$ reads:
\begin{equation}
\label{entropyseries}
h(\vecX)\approx-\ln m+\sum_{a=1}^{C}c_a\int f^{a}_{\mathbf X}(\mathbf x)\mathrm{d}\mathbf{x},
\end{equation}
and by choosing $\{c_a\}$ appropriately, and analytically evaluating the integral, we obtain an efficient approximation for $h(\mathbf X)$. We start by rewriting the integral in equation \eqref{entropyseries}, and the result is Corollary \ref{fpowerintegralcorr}.
\begin{lemma}
\label{productgausslemma}
Let $\{g_i(\vecx)\}_{i=1}^{q}$ be given in accordance with equation \eqref{gaussmix}, $\hat{t}$ s.t. $\sum_{i=1}^{q}t_i=a$ for some $a$. Furthermore, let ${M}^{-1}$, and ${\bm{\mu}}$ be given by:
\begin{align}
&{M}^{-1}=\sum_{j=1}^{q}t_{j}K_j^{-1},\\
&{\bm{\mu}}=\sum_{j=1}^{q}t_{j}{M} K_j^{-1}\mathbf{w}_j,
\end{align}
then we have:
\begin{equation}
\int \prod_{j=1}^{q}g^{t_j}_j(\mathbf x)\mathrm{d}\mathbf{x}=D(\hat{t}),
\end{equation}
where $D(\hat{t})$ is given by:
\begin{equation}
\label{Deq}
{D}(\hat t)=(2\pi)^{-n({a}-1)/2}\left(\prod_{i=1}^{q}(\det K_i)^{-t_{i}/2}\right)(\det {M})^{1/2}\exp\left\{-\dfrac{1}{2}\left(\sum_{l=1}^{q}t_{l}\mathbf{w}^{\mathrm{T}}_l K^{-1}_l \mathbf{w}_l-{\bm{\mu}}^{\mathrm{T}} {M}^{-1} {\bm{\mu}}\right)\right\}.
\end{equation}
\end{lemma}
\begin{proof}
See Appendix \ref{productgausslemmaapp}.
\end{proof}
It follows that the integral in equation \eqref{entropyseries} can be rewritten as shown in the following corollary.
\begin{corollary}
\label{fpowerintegralcorr}
The differential entropy $h(\vecX)$ of the Gaussian mixture is approximated by:
\begin{equation}\label{fpowerintegral}
\bar{h}_{\vec{c},m}(\vecX)=-\ln m+\sum_{a=1}^{C}c_a\sum_{{\substack{t_1, t_2,\dots, t_q \geq 0\\t_1+t_2+\dots+t_q=a}}} \binom{a}{t_1,t_2,\dots,t_q}\left(\prod_{i=1}^{q}p_i^{t_i}\right){D}(\hat t).
\end{equation}
\end{corollary}
\begin{proof}
See Appendix \ref{fpowerintegralcorrapp}.
\end{proof}
\subsection{Taylor series approximation}
\label{Taylorsec}
We perform a Taylor series for the logarithm only, in order to obtain the coefficients $\{c_a\}$ in equation \eqref{polyfit0}. We define $Z \equiv 1-m^{-1}f_{\mathbf X}(\mathbf X)$ and we write:
\begin{equation}
\label{flnftaylor}
\begin{aligned}
-f_{\vecX}\ln \dfrac{f_{\vecX}}{m}=-f_{\vecX}\ln (1-Z)=f_{\vecX}\sum_{k=1}^{\infty}\dfrac{1}{k}Z^k.
\end{aligned}
\end{equation}
The range of allowed $m$ values is the one that keeps $Z$ in the radius of convergence of the Taylor series: $|Z|<1$, i.e. $m \geq \max_{\mathbf x \in \mathbb R^n}f_{\vecX}(\mathbf{x})/2$. For $m= \max_{\mathbf x \in \mathbb R^n}f_{\vecX}(\mathbf{x})/2$ then $Z\in [-1,1)$, while for $m= \max_{\mathbf x \in \mathbb R^n}f_{\vecX}(\mathbf{x})$ then $Z \in [0,1)$. Furthermore, values above $\max_{\mathbf x \in \mathbb R^n}f_{\vecX}(\mathbf{x})$ such as $m=\sum_i p_ig_i(\mathbf{w}_i)$ shrink the range of $Z$ from below towards 1. By performing the Taylor expansion for the logarithm, we obtain a series of the form given by equation \eqref{polyfit0} as shown in Theorem \ref{taylortheo}.
\begin{theorem}
\label{taylortheo}
The procedure of performing the Taylor expansion of the logarithm as detailed in equation \eqref{flnftaylor} yields the following coefficients $\{c_a^{\mathrm{Taylor}}\}_{a=1}^{C}$ of the corresponding polynomial approximation in equation \eqref{polyfit0}:
\begin{equation}
c_{a}^{\mathrm{Taylor}}=\begin{cases}
                             H_{C-1}, & a=1, \\
          \dfrac{(-1)^{a+1}}{m^{a-1}}\dfrac{1}{a-1}\binom{C-1}{a-1}, & \mathrm{otherwise}.\end{cases},
\end{equation}
and the corresponding differential entropy approximation as given by equation \eqref{fpowerintegral} can be written~as:
\begin{equation}
\bar{h}^{\mathrm{Taylor}}_{C,m}(\vecX)=-\ln m+H_{C-1}-m\sum_{a=1}^{C-1}\dfrac{B_{a+1}}{a}\binom{C-1}{a},
\end{equation}
where $H_k$ is the $k$-th Harmonic number, and $\{B_a\}$ are given by:
\begin{equation}
{B}_{{a}}=\dfrac{(-1)^{{a}}}{m^{{a}}}\sum_{{\substack{t_1+t_2+\dots+t_q={a}\\t_1, t_2,\dots, t_q \geq 0}}}\binom{{a}}{t_1,t_2,\dots,t_q}\left(\prod_{i=1}^{q}p_i^{t_i}\right){D}(\hat{t}).
\end{equation}
\end{theorem}
\begin{proof}
See Appendix \ref{taylortheorapp}
\end{proof}
Note that one can show that $h(\vecX)$ is given by the infinite series:
\begin{equation}
h(\vecX)=-\ln m-m\sum_{k=1}^{\infty}\sum_{a=0}^{k}\dfrac{{B}_{a+1}}{k}\binom{k}{a},
\end{equation}
which is not the result of an infinite power series in $f^{a}(\vecX)$ as evident by the $a$ summation which is inside the $k$ summation. Taking the first $C-1$ terms and changing the order of summations makes the coefficients $\{c_{a}^{\mathrm{Taylor}}\}$ dependent on $C$. Note that Theorem \ref{taylortheo} with $m=\max_{\mathbf x \in \mathbb R^n}f_{\vecX}(\mathbf{x})$ provides a lower bound for $h(\vecX)$. For some applications having a bound is more important than having a good approximation.
\subsection{Polyfit}
\label{polyfitsec}
\subsubsection{General Polyfit}
Suppose we have a function $f(s)$ defined on $\mathcal{I}=[a,b]$ that we wish to approximate via a polynomial expansion. We write:
\begin{equation}
\label{generalpolyfit}
f(s)\approx\hat{f}(s)=\sum_{i=1}^{C}d_{\mathcal{I},i}s^{i},
\end{equation}
and we define the error resulting from our estimation as:
\begin{equation}
\label{errorfun}
E=\dfrac{1}{b-a}\int_{a}^{b} w(s)(f(s)-\hat{f}(s))^2\mathrm{d}s,
\end{equation}
where we consider the possibility of favoring portions of $\mathcal{I}$ more than others by choosing the weight function $w(s)$ appropriately, and $w(s)=1$ reduces $E$ to the mean squared error. Our choice of $w$ and $\mathcal{I}$ must be such that $M_{\mathcal{I}}$ in Lemma \ref{generalpolyfitlemma} is invertible.
\begin{lemma}
\label{generalpolyfitlemma}
Let $M_{\mathcal{I}}$, $\vec{d}_{\mathcal{I}}$, and $\vec{z}_{\mathcal{I}}$ be given by:
\begin{align}
&(M_{\mathcal{I}})_{ij} \equiv \int_{a}^{b} w(s)s^{i+j}\mathrm{d}s,\\
&(\vec{d}_{\mathcal{I}})_i\equiv d_{\mathcal{I},i},\\
&(\vec{z}_{\mathcal{I}})_i\equiv \int_{a}^{b} w(s)f(s)s^{i}\mathrm{d}s,
\end{align}
then for invertible $M_{\mathcal{I}}$, the coefficients $d_{\mathcal{I},i}$ in equation \eqref{generalpolyfit} that minimize the error function $E$ in equation \eqref{errorfun} are given by:
\begin{equation}
\label{dmatrixeq}
\vec{d}_{\mathcal{I}}=M^{-1}_{\mathcal{I}}\vec{z}_{\mathcal{I}}.
\end{equation}
\end{lemma}
\begin{proof}
See Appendix \ref{generalpolyfitlemmaproof}.
\end{proof}
\subsubsection{Entropy estimate obtained from Polyfit}
\label{entropypolyfit}
We now wish to use Lemma \ref{generalpolyfitlemma} to obtain an approximation for $h(\vecX)$. We first apply our estimator to the function $f(s)=-s\ln s$ in the following lemma.
\begin{lemma}
\label{slogslemma}
Let $f(s)=-s\ln s$ on $\mathcal{I}=(a,b]$, $r \in \mathbb{R}$ if $a>0$ and $r>-3$ if $a=0$. Let $w(s)=s^r$, then the estimator $\hat{f}(s)$ in equation \eqref{generalpolyfit} is given by:
\begin{equation}
-s\ln s \approx \sum_{i=1}^{C}d_{\mathcal{I},i}s^{i},
\end{equation}
where $d_{\mathcal{I},i}$ are obtained by solving equation \eqref{dmatrixeq} with $M_{\mathcal{I}}$ and $\vec{z}_{\mathcal{I}}$ given by:
\begin{align}
&(M_{\mathcal{I}})_{ij} = \dfrac{b^{i+j+r+1}-a^{i+j+r+1}}{i+j+r+1},\\
\nonumber
&(\vec{z}_{\mathcal{I}})_i=\dfrac{1}{(i+r+2)^2}\left[b^{i+r+2}(1-(i+r+2)\ln b)-a^{i+r+2}(1-(i+r+2)\ln a)\right].
\end{align}
\end{lemma}
\begin{proof}
It follows from equation \eqref{generalpolyfit}, Lemma \ref{generalpolyfitlemma}, and direct calculation.
\end{proof}
\begin{corollary}
\label{Mtildecorr}
Let $f(s)=-s\ln s$ on $\mathcal{I}=(0,b]$, $w(s)=s^r$ with $r >-3$, then $\{d_{\mathcal{I},i}\}$ are given by:
\begin{equation}\label{dtodtilde}
d_{\mathcal{I},i}=\begin{cases}
                    \tilde{d}_1-\ln b, & \mathrm{if \,} i=1 \\
                    b^{1-i}\tilde{d}_i, & \mathrm{{otherwise}}
                  \end{cases},
\end{equation}
where $\vec{\tilde{d}}$ is obtained by solving:
\begin{equation}
\vec{\tilde{d}}=\tilde{M}^{-1}\vec{\tilde{z}},
\end{equation}
with $\tilde{M}$ and $\vec{\tilde{z}}$ given by:
\begin{equation}
\tilde{M}_{ij} = \dfrac{1}{i+j+r+1},\,\,\,\,\,\tilde{z}_i=\dfrac{1}{(i+r+2)^2}.
\end{equation}
\end{corollary}
\begin{proof}
See Appendix \ref{Mtildecorrproofapp}.
\end{proof}
Corollary \ref{Mtildecorr} shows that the problem of finding a polynomial fit for $f(s)=-s\ln s$ on $(0,b]$ is equivalent to finding the inverse of $\tilde{M}$, which is independent of $b$. However, $\tilde{M}$ is an ill-conditioned matrix as the elements $\tilde{M}_{ij}$ become increasingly smaller for high $i+j$. For example, if we take $r=-2$ then $\tilde{M}$ is the Hilbert matrix with inverse elements: $(H^{-1})_{ij} = (-1)^{i+j} (i + j - 1) \binom{C+i-1}{C-j} \binom{C+j-1}{C-i} \left[ \binom{i+j-2}{i-1} \right]^2$, which become numerically unstable at $C \sim 10$. If we write $f_{\mathrm{max}}\equiv \max_{\mathbf x \in \mathbb R^n}f_{\vecX}(\mathbf{x})$, then for our Gaussian mixture $f_\vecX$, the possible values are in $(0,f_{\mathrm{max}}]$. In a very large part of $\mathbb{R}^n$ we find ourselves in the tails of the Gaussian mixture, which corresponds to $f_\vecX(\vecx)$ being close to zero. Hence, we want our weight function $w(s)$ to reflect this, and emphasize the region close to zero in the application of Corollary \ref{Mtildecorr}. Note that we can write:
\begin{equation}
\label{entropyins}
\begin{aligned}
&h(\vecX)=-\int f_{\vecX}(\vecx)\ln f_{\vecX}(\vecx) \mathrm{d}\vecx=-\int \left(\int_{0}^{f_{\mathrm{max}}}\delta(s-f_{\vecX}(\vecx))s\ln s\,\mathrm{d}s\right) \mathrm{d}\vecx\\
&=-\int_{0}^{f_{\mathrm{max}}} s\ln s\left(\int\delta(s-f_{\vecX}(\vecx))\mathrm{d}\vecx\right)\mathrm{d}s=\int_{0}^{f_{\mathrm{max}}} f(s)V(s)\,\mathrm{d}s,
\end{aligned}
\end{equation}
where $V(s)=\int\delta(s-f_{\vecX}(\vecx))\mathrm{d}\vecx$ can be thought of as the $(n-1)$-dimensional volume where the function $f_{\vecX}$ equals $s$. We see from equation \eqref{entropyins} that our entropy approximation is not only determined by how well we can estimate $f(s)$ using Corollary \ref{Mtildecorr}, but also by $V(s)$, which for Gaussian mixtures will put more weight around $s=0$. That is, we expect $r<0$ in Corollary \ref{Mtildecorr} to produce better results for our entropy approximation than $r>0$, albeit they will produce worse estimates of $f(s)$ over any interval. Ideally we set $w(s)=V(s)$ in Lemma \ref{slogslemma}, but $V(s)$ has no closed form solution. This leads us to the following entropy approximation given by Theorem \ref{polyfittheo}.
\begin{theorem}
\label{polyfittheo}
Let $r >-3$, $\mathcal{I}=(0,f_{\mathrm{max}}]$, then the entropy approximation $\bar{h}^{\mathrm{Polyfit}}_{C}(\vecX)$ based on Lemma \ref{slogslemma}, Corollary \ref{Mtildecorr}, and Corollary \ref{fpowerintegralcorr} is given by:
\begin{equation}
\bar{h}^{\mathrm{Polyfit}}_{C}(\vecX)=\sum_{a=1}^{C}c^{\mathrm{Polyfit}}_{a}\sum_{{\substack{t_1+t_2+\dots+t_q=a\\t_1, t_2,\dots, t_q \geq 0}}} \binom{a}{t_1,t_2,\dots,t_q}\left(\prod_{i=1}^{q}p_i^{t_i}\right){D}(\hat t),
\end{equation}
where $\vec{c}^{\mathrm{\,Polyfit}}$ is given by:
\begin{equation}
c_{a}^{\mathrm{Polyfit}}=\begin{cases}
                    \tilde{d}_1-\ln f_{\mathrm{max}}, & \mathrm{if \,} a=1 \\
                    f_{\mathrm{max}}^{1-a}\tilde{d}_a, & \mathrm{{otherwise}}
                  \end{cases},\quad \vec{\tilde{d}}=\tilde{M}^{-1}\vec{\tilde{z}},
\end{equation}
with $\tilde{M}_{ij}$ and ${\tilde{z}}_i$ given by:
\begin{equation}
\tilde{M}_{ij} = \dfrac{1}{i+j+r+1},\quad \tilde{z}_i=\dfrac{1}{(i+r+2)^2}.
\end{equation}
\end{theorem}
\begin{proof}
It follows from the direct application of Lemma \ref{slogslemma}, Corollary \ref{Mtildecorr}, and Corollary \ref{fpowerintegralcorr} to $-f_{\vecX}\ln f_{\vecX}$.
\end{proof}
In Appendix \ref{app:volume}, we derive the volume function $V(s)$ for the trivial case $q=1$. Around $s\approx 0$ it blows up as $1/s$, with some additional logarithmic divergence; this gives some intuition why we get good results at $r\leq -1$.
\section{Numerical results}
\label{resultsec}
\subsection{Polyfit results}
We now present results of Theorem \ref{polyfittheo} applied to different Gaussian mixtures in order to estimate $h(\mathbf X)$. The performance metric we use is the percentage error compared to the exact value for the entropy. The mixtures we used to show the validity of our approximation are shown in Table \ref{tab}. The examples cover various values of $n$ and $q$. This is not meant to be all encompassing, but still covers a sufficient part of the parameter space to showcase the accuracy of our approximation.
\begin{table}[H]
\centering
{\small
\begin{longtable}{|c|c|p{14.5cm}|}
     \caption{\it The different Gaussian mixtures tested using Theorem
\ref{polyfittheo}. The results for the entropy are shown in Figure
\ref{polyfitfig}.}
     \label{tab}
     \\
     \hline
     $q$ & $n$ & \textbf{Mixture Parameters} \\
     \hline \hline
     3 & 2 &
     $\mathbf{w}_1 = \begin{pmatrix} 1 & 0 \end{pmatrix}^{\mathrm{T}}$,
\vphantom{$M^{M^{2^2}}$}
     $\mathbf{w}_2 = \begin{pmatrix} -1 & 0 \end{pmatrix}^{\mathrm{T}}$,
     $\mathbf{w}_3 = \begin{pmatrix} 0 & 1.5
\end{pmatrix}^{\mathrm{T}}$, \quad
     $K_i = I_2$, \quad
     $\mathbf{p} = \begin{pmatrix} 0.2 & 0.3 & 0.5
\end{pmatrix}^{\mathrm{T}}$ \\
     \cline{2-3}
     & 2 &
     $\mathbf{w}_1 = \begin{pmatrix} 0 & 0 \end{pmatrix}^{\mathrm{T}}$,
\vphantom{$M^{M^{2^2}}$}
     $\mathbf{w}_2 = \begin{pmatrix} -1.5 & 1.5 \end{pmatrix}^{\mathrm{T}}$,
     $\mathbf{w}_3 = \begin{pmatrix} 1.5 & 1.5 \end{pmatrix}^{\mathrm{T}}$,
     $K_1 = \begin{pmatrix} 1 & 0 \\ 0 & 3 \end{pmatrix}$,
     $K_2 = \begin{pmatrix} 1 & 0.2 \\ 0.2 & 1 \end{pmatrix}$,
     $K_3 = \begin{pmatrix} 1 & -0.2 \\ -0.2 & 1 \end{pmatrix}$, \quad\quad
     $\mathbf{p} = \begin{pmatrix} 0.2 & 0.3 & 0.5
\end{pmatrix}^{\mathrm{T}}$ \\ \hline\hline

     4 & 3 &
     $\mathbf{w}_1 = \begin{pmatrix} 0 & 0 & 0
\end{pmatrix}^{\mathrm{T}}$, \vphantom{$M^{M^{2^2}}$}
     $\mathbf{w}_2 = \begin{pmatrix} -1.5 & 1.5 & -1.5
\end{pmatrix}^{\mathrm{T}}$,
     $\mathbf{w}_3 = \begin{pmatrix} 1.5 & 1.5 & 1.5
\end{pmatrix}^{\mathrm{T}}$,
     $\mathbf{w}_4 = \begin{pmatrix} 1 & 1 & 1 \end{pmatrix}^{\mathrm{T}}$,
     $K_i = I_3$, \quad\quad
     $\mathbf{p} = \begin{pmatrix} 0.2 & 0.3 & 0.3 & 0.2
\end{pmatrix}^{\mathrm{T}}$ \\ \cline{2-3}
     & 8 &
     $\mathbf{w}_1 = \begin{pmatrix} 0 & 0 & 0 & 0 & 0 & 0 & 0 & 0
\end{pmatrix}^{\mathrm{T^{\vphantom{2}} }}$,
     $\mathbf{w}_2 = \begin{pmatrix} -1.5 & 1.5 & -1.5 & 1.5 & -1.5 &
1.5 & -1.5 & 1.5 \end{pmatrix}^{\mathrm{T}}$,
     $\mathbf{w}_3 = \begin{pmatrix} 1.5 & 1.5 & 1.5 & 1.5 & 1.5 & 1.5 &
1.5 & 1.5 \end{pmatrix}^{\mathrm{T}}$,
     $\mathbf{w}_4 = \begin{pmatrix} 1 & 1 & 1 & 1 & 1 & 1 & 1 & 1
\end{pmatrix}^{\mathrm{T}}$, \newline
     $K_i = I_8$, \quad\quad
     $\mathbf{p} = \begin{pmatrix} 0.2 & 0.3 & 0.3 & 0.2
\end{pmatrix}^{\mathrm{T}}$ \\ \hline\hline

     5 & 4 &
     $\mathbf{w}_1 = \begin{pmatrix} 0 & 0 & 0 & 0
\end{pmatrix}^{\mathrm{T}}$, \vphantom{$M^{M^{2^2}}$}
     $\mathbf{w}_2 = \begin{pmatrix} -1.5 & 1.5 & -1.5 & 1.5
\end{pmatrix}^{\mathrm{T}}$,
     $\mathbf{w}_3 = \begin{pmatrix} 1.5 & 1.5 & 1.5 & 1.5
\end{pmatrix}^{\mathrm{T}}$, \newline
     $\mathbf{w}_4 = \begin{pmatrix} 1 & 1 & 1 & 1
\end{pmatrix}^{\mathrm{T}}$,
     $\mathbf{w}_5 = \begin{pmatrix} -3 & 3 & -3 & 3
\end{pmatrix}^{\mathrm{T}}$, \quad
     $K_i = I_4$, \quad
     $\mathbf{p} = \begin{pmatrix} 0.2 & 0.3 & 0.3 & 0.1 & 0.1
\end{pmatrix}^{\mathrm{T}}$ \\ \hline
\end{longtable}
}
\end{table}
From Figure \ref{polyfitfig}, we see that for low dimensional cases ($q=3, n=2$), our approximation converges quickly to the exact value of the entropy (within less than $1\%$ error), and we can fairly trust numerical accuracy for those cases.
As for the optimal $r$ value for the weight function, $r=-2$ is consistently the best choice for those parameters, with $r>0$ being less favorable, as predicted in Section \ref{entropypolyfit}.
At $q=4$, we start to see divergent behavior as $n$ becomes larger.
For the lower dimensional case of $q=4, n=3$, the behavior is similar to previous cases with very high accuracy and $r=-2$ being the most favorable weight function.
However, in the case of $q=4, n=8$, the result diverges as $C$ grows larger, although up to $C=7$ the approximation is fairly accurate.
Furthermore, the best performing $r$ in this case is $r=-2.5$ with an error percentage of around $-0.5\%$ at $C=6$.
We conjecture that the divergent behavior in this case is caused by numerical inaccuracies in the evaluation of multinomial coefficients.
\begin{figure}[H]
    \centering
    \begin{subfigure}[b]{0.43\linewidth}
        \centering
        \includegraphics[width=0.9\linewidth]{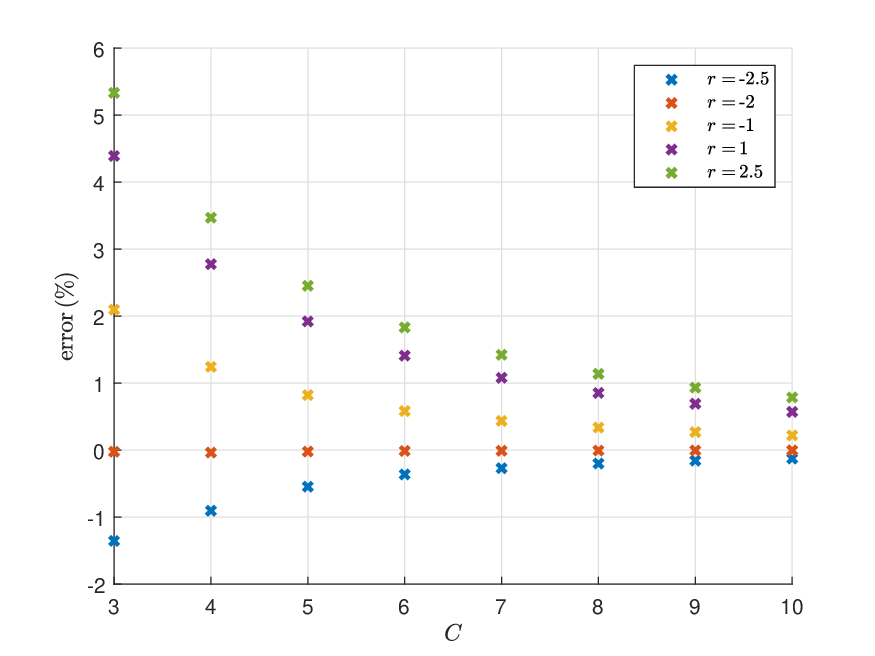}
        \caption{$q=3, n=2, K_i=I_2$}
        \label{fig:case1_graph1}
    \end{subfigure}
    \hspace{0.02\linewidth} 
    \begin{subfigure}[b]{0.43\linewidth}
        \centering
        \includegraphics[width=0.9\linewidth]{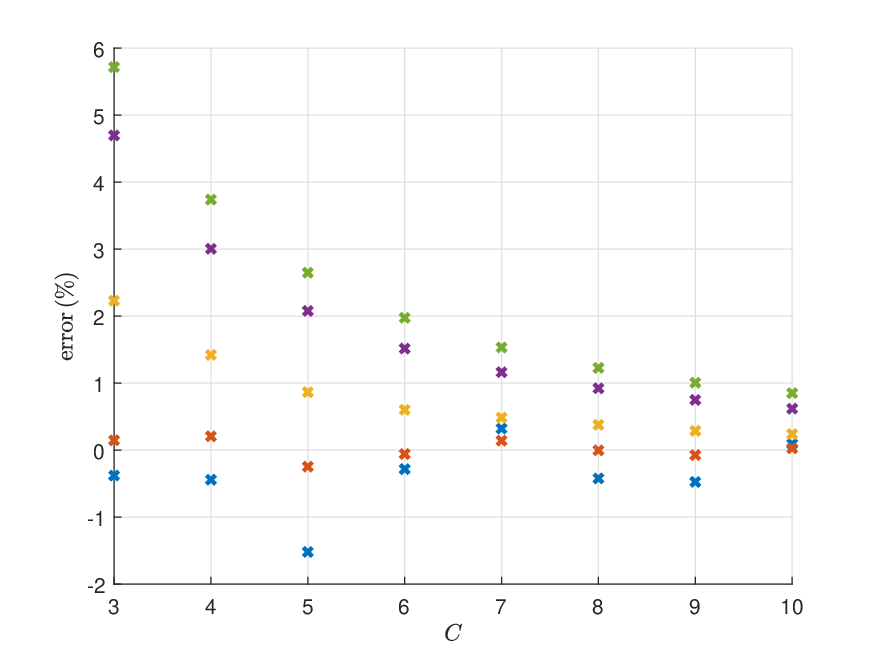}
        \caption{$q=3, n=2$}
        \label{fig:case1_graph2}
    \end{subfigure}

    \vspace{-0em} 

    \begin{subfigure}[b]{0.43\linewidth}
        \centering
        \includegraphics[width=0.9\linewidth]{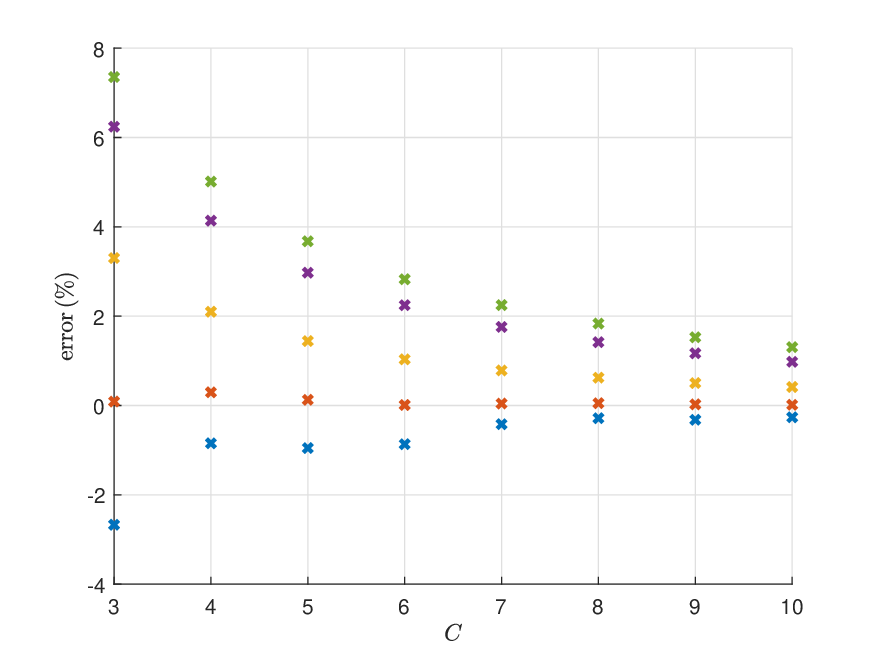}
        \caption{$q=4, n=3$}
        \label{fig:case2_graph1}
    \end{subfigure}
    \hspace{0.02\linewidth} 
    \begin{subfigure}[b]{0.43\linewidth}
        \centering
        \includegraphics[width=0.9\linewidth]{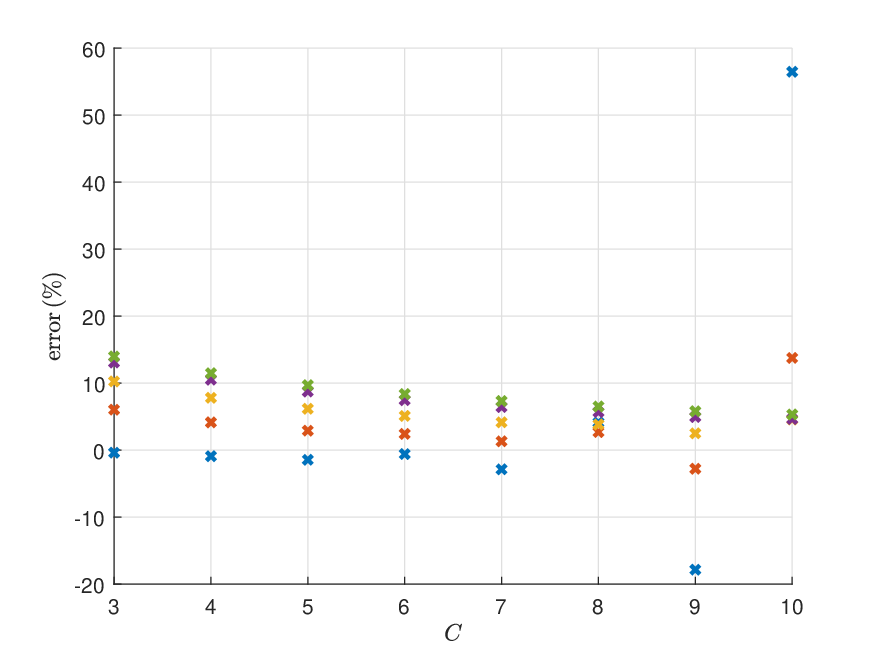}
        \caption{$q=4, n=8$}
        \label{fig:case2_graph2}
    \end{subfigure}

    \vspace{-0em} 

    \begin{subfigure}[b]{0.43\linewidth}
        \centering
        \includegraphics[width=0.9\linewidth]{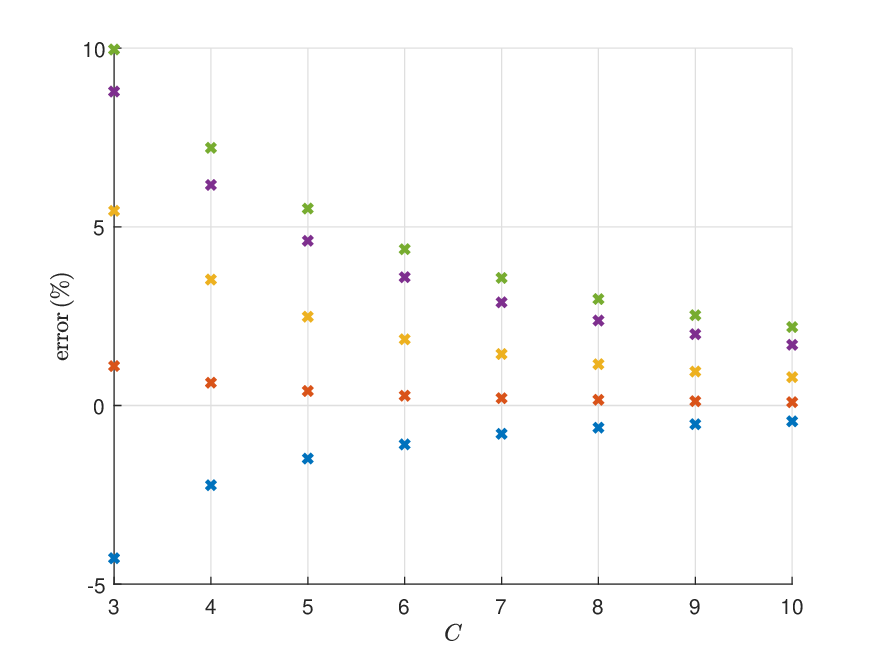}
        \caption{$q=5, n=4$}
        \label{fig:case3_graph1}
    \end{subfigure}

    \caption{The results of testing Theorem \ref{polyfittheo} on the Gaussian mixtures in Table \ref{tab}. The graphs show the percentage error $\left({\textstyle\frac{h(\vecX)-\bar{h}^{\mathrm{Polyfit}}_{C}(\vecX)}{h(\vecX)}}100\right)$ of the approximation $\bar{h}^{\mathrm{Polyfit}}_{C}(\vecX)$ at different values of $C$.}
    \label{polyfitfig}
\end{figure}
For $q=5,n=4$, we again have convergent behavior yielding an accurate approximation of the entropy with $r=-2$ being the best weight function.
This prompts us to conclude that our method produces a very highly accurate approximation when $r \in [-2.5,-2]$
and $C$ is kept between 3 and 8 in order to avoid computational inaccuracies. Note that the upper bound mentioned in the introduction $h(\vecX) \leq \sum_{j=1}^{q} p_j\ln \frac1{p_j} + \sum_{j=1}^{q} p_j\frac 12\ln\det (2\pi e K_j)$ produces values that are within $3-19 \%$ above the true entropy for the configurations in Table \ref{tab}. In Figure \ref{funapproxfig}, we show Polyfit approximations of $f(s)=-s \ln s$ as described in Lemma \ref{slogslemma} for different values of $r$. We see clearly in Figure \ref{funapproxfig} that for $r=-2$ the approximation is overall inaccurate compared to $r=-1$ and $r=1$. However, $r=-2$ yields the best approximation for the entropy $h(\vecX)$. As we discussed in Section \ref{entropypolyfit}, the reason for this is that the value of $r=-2$ makes the approximation better around $s \approx 0$, which is where $V(s)$ is highest for Gaussian mixtures.
\begin{figure}[H]
    \centering
        \centering
        \includegraphics[width=0.6\linewidth]{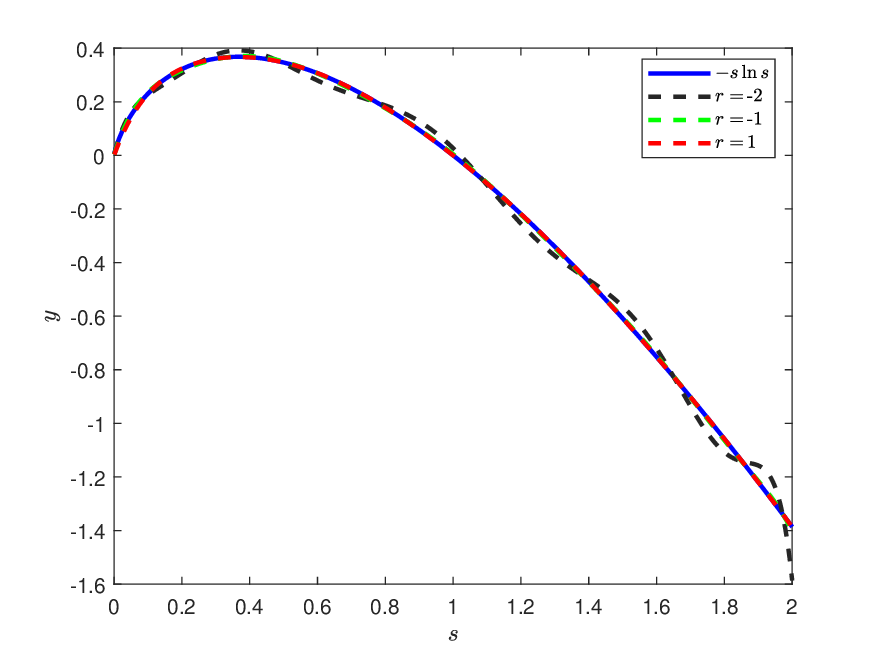}
        \caption{Polyfit approximations of $f(s)=-s \ln s$ for different values of $r$ on the interval (0,2] in accordance with Lemma \ref{slogslemma}.}
        \label{funapproxfig}
\end{figure}
\subsection{Taylor series results}
We now show results for our approximation $\bar{h}^{\mathrm{Taylor}}_{C,m}(\vecX)$ based on the Taylor expansion of the logarithm as outlined in Theorem \ref{taylortheo}. We do not expect it to be as accurate as $\bar{h}^{\mathrm{Polyfit}}_{C,m}(\vecX)$, however, it provides us with an easy to compute lower bound for the entropy, as mentioned earlier in Section \ref{Taylorsec}. We show the percentage error compared to the exact value of the entropy for different values of $\beta$ for two mixtures that we take from Table \ref{tab}. We use the case of $q=3, n=2$ with non-spherical covariance matrices, as well as the case of $q=4, n=8$. The results are shown in Figure \ref{taylorfig}. As we can see from Figure \ref{taylorfig}, $\bar{h}^{\mathrm{Taylor}}_{C,m}(\vecX)$ is not as accurate as $\bar{h}^{\mathrm{Polyfit}}_{C,m}(\vecX)$ as expected, especially in the higher dimensional case. However, it does provide monotonous and convergent behavior. Furthermore, the case $\beta=1/2$ is the fastest converging case, which is also the case where computational inaccuracies start to appear the fastest as we increase $C$.
\begin{figure}[H]
    \centering
    \begin{subfigure}[b]{0.46\linewidth}
        \centering
        \includegraphics[width=\linewidth]{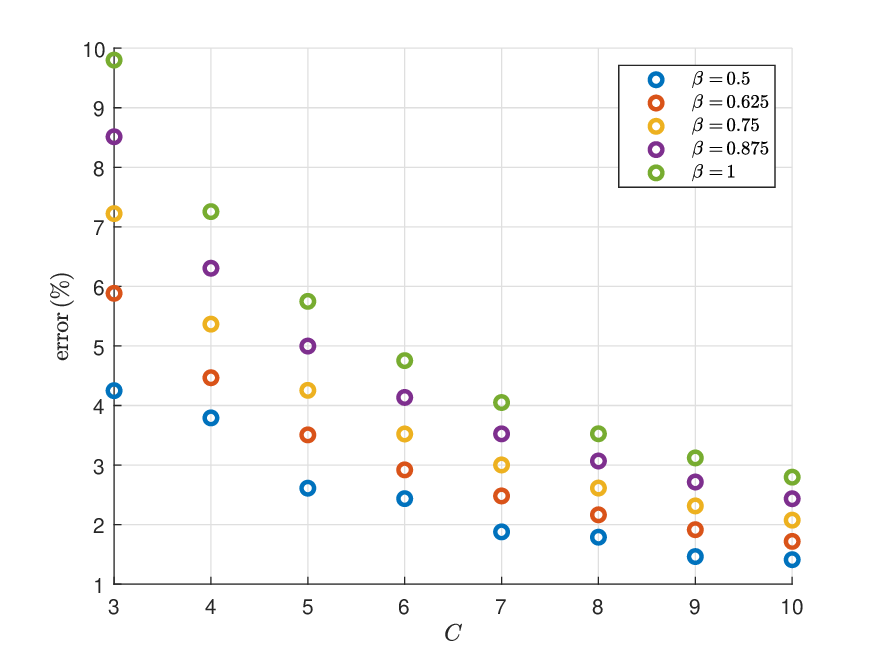}
        \caption{$q=3, n=2$}
        \label{fig:case1_graph1 taylor}
    \end{subfigure}
    \hspace{0.02\linewidth} 
    \begin{subfigure}[b]{0.46\linewidth}
        \centering
        \includegraphics[width=\linewidth]{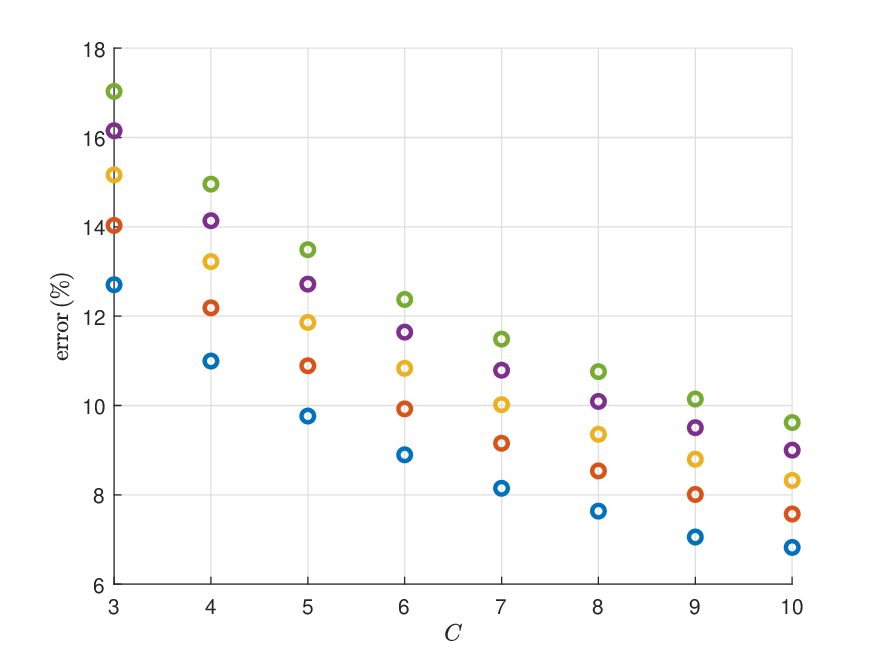}
        \caption{$q=4, n=8$}
        \label{fig:case1_graph2 taylor}
    \end{subfigure}
    \caption{The results of testing Theorem \ref{taylortheo} on two of the Gaussian mixtures in Table \ref{tab}. The graphs show the percentage error $\left({\textstyle\frac{h(\vecX)-\bar{h}^{\mathrm{Taylor}}_{C,m}(\vecX)}{h(\vecX)}}100\right)$ of the approximation $\bar{h}^{\mathrm{Taylor}}_{C,m}(\vecX)$ at different values of $C$ and for different values of $\beta$.}
    \label{taylorfig}
\end{figure}
\section{Discussion}
\label{discsec}
Gaussian mixtures are one of the most applicable distributions in the literature, from wireless communications to physics and diffusion models. They have found success in either describing or modeling real life phenomena to a high degree of accuracy. One of the most important quantities in characterizing a distribution is the Shannon (differential) entropy, which for Gaussian mixtures, does not have a known closed form expression. In this paper, we have derived a method by which the differential entropy of Gaussian mixtures can be approximated both accurately and efficiently. Our method enjoys a high degree of simplicity as it relies on the fact that for Gaussian mixtures \emph{``almost all the support has an image that is almost zero''}. For various configurations, this trick allows our approximation to get very close to the true value for the differential entropy by summing only a small number of terms. Although some configurations are more difficult to approximate with a small number of terms, our method is accurate for a large number of randomly generated Gaussian mixtures (not reported in Section \ref{resultsec}). This makes inaccuracy an exception that can be circumvented by avoiding numerical instabilities and taking into account higher order terms.

It is not straightforward to benchmark against the existing literature. Many methods are at least partly numerical, whereas our approximation technique yields fully analytic expressions. The work \cite{approx} is closest to our technique in that it approximates $h(\vecX)$ with closed-form analytic expressions. However, there are significant qualitative differences. Since \cite{approx} tries to polynomially approximate $\ln f(\vecx)$, the order $C$ of their Taylor series has to increase with $q$ as $C=2q$ (in the worst case) in order to correctly fit all the `bumps' in the function. The error in $\ln f(\vecx)$ then explodes as $|\vecx|^{2q}$ when $|\vecx|$ increases beyond a certain radius. This large error should be damped by the fact that it gets multiplied times the Gaussian tail of $f(\vecx)$. However, at large $q$ this can require drastic application of the splitting trick, which introduces errors. In contrast, our choice of polynomial order $C$ is independent of~$q$, and increasing $q$ does not lead to errors. In our technique, \textcolor[rgb]{0,0,0}{inaccuracies are caused by the mismatch between the function $s\ln\frac1s$ and the polynomial fit around $s=0$. The performance of our technique is constrained by the fact that the function $s\ln\frac1s$ has infinite derivative at $s=0$ and hence cannot be fitted with a polynomial of finite degree at that point.} Our Polyfit method relies on choosing an appropriate weight function $w(s)$ as discussed in Section \ref{polyfitsec}. Ideally for Gaussian mixtures, $w(s)$ should be optimally chosen as $V(s)$ which is difficult to estimate. However, from our experiments we notice that a simple power $w(s)\propto s^r$ gives good results, especially for negative $r$ around $r=-2$. At negative $r$, a lot of weight is given to small values of~$s$. This makes sense because most of the $n$-dimensional volume in the integral $\int\!\rd\vecx$ covers the tails of the distribution $f(\vecX)$. For a single Gaussian with spherical symmetry it is easy to verify (see Appendix~\ref{app:volume}) that the volume function
$V(s)=\int\rd\vecx\; \delta(s-f(\vecx))$ behaves as
$V(s)\propto\frac1s (\ln\frac{f_{\rm max}}{s} )^{\frac n2-1}$. At small $s$ this expression blows up faster than $\frac1s$ but slower than $\frac1{s^2}$. This suggests that $r$ should lie in the vicinity of $[-2,-1]$.

More extensive testing, with larger mixtures in higher dimensions,
will show how generally applicable our method is. We have not exhaustively studied all possible polynomial fits.
It is quite likely that some improvement can be gained by choosing
a better weight function, e.g.\;one that more resembles $V(s)$. 


\subsection*{Acknowledgements}
This work was supported by NWO grant CS.001 (Forwardt) and by the Dutch Groeifonds project Quantum Delta NL KAT-2.
\bibliographystyle{plain}
\bibliography{references}

\begin{thebibliography}{10}

\bibitem{relat2}
Jacob Goldberger, Shiri Gordon, and Hayit Greenspan.
\newblock {An efficient image similarity measure based on approximations of
  KL-divergence between two Gaussian mixtures}.
\newblock In {\em Proceedings Ninth IEEE International conference on computer
  vision}, pages 487--493. IEEE, 2003.

\bibitem{DIFFREF4}
Hanzhong Guo, Cheng Lu, Fan Bao, Tianyu Pang, Shuicheng Yan, Chao Du, and
  Chongxuan Li.
\newblock Gaussian mixture solvers for diffusion models.
\newblock {\em Advances in Neural Information Processing Systems}, 36, 2024.

\bibitem{relat}
John~R Hershey and Peder~A Olsen.
\newblock {Approximating the Kullback Leibler divergence between Gaussian
  mixture models}.
\newblock In {\em 2007 IEEE International Conference on Acoustics, Speech and
  Signal Processing-ICASSP'07}, volume~4, pages IV--317. IEEE, 2007.

\bibitem{DIFFREF2}
Jonathan Ho, Ajay Jain, and Pieter Abbeel.
\newblock Denoising diffusion probabilistic models.
\newblock {\em Advances in neural information processing systems},
  33:6840--6851, 2020.

\bibitem{approx}
Marco~F Huber, Tim Bailey, Hugh Durrant-Whyte, and Uwe~D Hanebeck.
\newblock {On entropy approximation for Gaussian mixture random vectors}.
\newblock In {\em 2008 IEEE International Conference on Multisensor Fusion and
  Integration for Intelligent Systems}, pages 181--188. IEEE, 2008.

\bibitem{JoudehSkoric}
B.~Joudeh and B.~\v{S}kori\'{c}.
\newblock {Average entropy of Gaussian mixtures}.
\newblock {\em Entropy}, 26(8):659, 2024.
\newblock Preprint arXiv:2404.07311.

\bibitem{BIOREF1}
Paul~D McNicholas and Thomas~Brendan Murphy.
\newblock Model-based clustering of microarray expression data via latent
  gaussian mixture models.
\newblock {\em Bioinformatics}, 26(21):2705--2712, 2010.

\bibitem{q2}
Joseph~V Michalowicz, Jonathan~M Nichols, and Frank Bucholtz.
\newblock Calculation of differential entropy for a mixed gaussian
  distribution.
\newblock {\em Entropy}, 10(3):200, 2008.

\bibitem{bound}
Frank Nielsen and Richard Nock.
\newblock A series of maximum entropy upper bounds of the differential entropy.
\newblock {\em arXiv preprint arXiv:1612.02954}, 2016.

\bibitem{CYBREF1}
Ashok Parmar, Karan Shah, Kamal Captain, Miguel L{\'o}pez-Ben{\'\i}tez, and
  Jignesh Patel.
\newblock Gaussian mixture model based anomaly detection for defense against
  byzantine attack in cooperative spectrum sensing.
\newblock {\em IEEE Transactions on Cognitive Communications and Networking},
  2023.

\bibitem{CYBREF2}
Xiaoying Qiu, Ting Jiang, Sheng Wu, and Monson Hayes.
\newblock Physical layer authentication enhancement using a gaussian mixture
  model.
\newblock {\em IEEE Access}, 6:53583--53592, 2018.

\bibitem{THERMREF1}
Neil Raymond, Dmitri Iouchtchenko, Pierre-Nicholas Roy, and Marcel Nooijen.
\newblock A path integral methodology for obtaining thermodynamic properties of
  nonadiabatic systems using gaussian mixture distributions.
\newblock {\em The Journal of Chemical Physics}, 148(19), 2018.

\bibitem{DIFFREF3}
Jeremias Sulam, Yaniv Romano, and Michael Elad.
\newblock Gaussian mixture diffusion.
\newblock In {\em 2016 IEEE International Conference on the Science of
  Electrical Engineering (ICSEE)}, pages 1--5. IEEE, 2016.

\bibitem{BIOREF2}
Hiroyuki Toh and Katsuhisa Horimoto.
\newblock Inference of a genetic network by a combined approach of cluster
  analysis and graphical gaussian modeling.
\newblock {\em Bioinformatics}, 18(2):287--297, 2002.

\bibitem{WIREF1}
Nurettin Turan, Benedikt B{\"o}ck, Kai~Jie Chan, Benedikt Fesl, Friedrich
  Burmeister, Michael Joham, Gerhard Fettweis, and Wolfgang Utschick.
\newblock Wireless channel prediction via gaussian mixture models.
\newblock {\em arXiv preprint arXiv:2402.08351}, 2024.

\bibitem{ASTROREF2}
Wynne Turner, Paul Martini, Naim~G{\"o}ksel Kara{\c{c}}ayl{\i}, J~Aguilar,
  S~Ahlen, D~Brooks, T~Claybaugh, A~de~la Macorra, A~Dey, P~Doel, et~al.
\newblock New measurements of the lyman-$\alpha$ forest continuum and effective
  optical depth with lycan and desi y1 data.
\newblock {\em arXiv preprint arXiv:2405.06743}, 2024.

\bibitem{DIFFREF1}
Yuchen Wu, Minshuo Chen, Zihao Li, Mengdi Wang, and Yuting Wei.
\newblock Theoretical insights for diffusion guidance: A case study for
  gaussian mixture models.
\newblock {\em arXiv preprint arXiv:2403.01639}, 2024.

\bibitem{ASTROREF1}
Hai Zhu, Rui Guo, Juntai Shen, Jianglai Liu, Chao Liu, Xiang-Xiang Xue, Lan
  Zhang, and Shude Mao.
\newblock The local dark matter kinematic substructure based on lamost k
  giants.
\newblock {\em arXiv preprint arXiv:2404.19655}, 2024.

\end{thebibliography}


\appendix

\section{Proof of Lemma \ref{productgausslemma}}
\label{productgausslemmaapp}
\small
\footnotesize{
\begin{equation}
\begin{aligned}
&\int \prod_{j=1}^{q}g^{t_{j}}_j(\mathbf x) \mathrm{d}{\mathbf x}=\int \prod_{j=1}^{q}\left((2\pi)^{-\frac n2}(\det K_j)^{-1/2}\exp \left\{- \dfrac{1}{2}(\vecx-\vecw_j)^\mathrm T K_j^{-1}(\vecx-\vecw_j)\right\}\right)^{t_{j}} \mathrm{d}{\mathbf x}\\
&=C(\hat{t})\int e^{-\frac{1}{2}\sum_{j=1}^{q}t_{j}(\vecx-\vecw_j)^\mathrm T K_j^{-1}(\vecx-\vecw_j)} \mathrm{d}{\mathbf x}\\
&=C(\hat{t})e^{- \frac{1}{2}\sum_{j=1}^{q}t_{j}\vecw_j^\mathrm T K_j^{-1}\vecw_j}\int\exp \left\{- \dfrac{1}{2}\sum_{j=1}^{q}t_{j}\left(\vecx^\mathrm T K_j^{-1}\vecx-\vecx^\mathrm T K_j^{-1}\vecw_j-\vecw_j^\mathrm T K_j^{-1}\vecx\right)\right\} \mathrm{d}{\mathbf x}\\
&=C(\hat{t})e^{ -\frac{1}{2}\sum_{j=1}^{q}t_{j}\vecw_j^\mathrm T K_j^{-1}\vecw_j}\int\exp \left\{- \dfrac{1}{2}\left(\vecx^\mathrm T \left(\sum_{j=1}^{q}t_{j}K_j^{-1}\right)\vecx-\vecx^\mathrm T \left(\sum_{j=1}^{q}t_{j}K_j^{-1}\vecw_j\right)-\left(\vecw_j^\mathrm T \sum_{j=1}^{q}t_{j}K_j^{-1}\right)\vecx\right)\right\} \mathrm{d}{\mathbf x}\\
&=C(\hat{t})e^{ -\frac{1}{2}\sum_{j=1}^{q}t_{j}\vecw_j^\mathrm T K_j^{-1}\vecw_j}\int\exp \left\{- \dfrac{1}{2}\left(\vecx^{\mathrm T} M^{-1}\vecx-\vecx^{\mathrm T} {M}^{-1}{\bm{\mu}}-{\bm{\mu}}^{\mathrm{T}}{M}^{-1}\vecx\right)\right\} \mathrm{d}{\mathbf x}\\
&=C(\hat{t})e^{-\frac{1}{2}\sum_{j=1}^{q}t_{j}\vecw_j^\mathrm T K_j^{-1}\vecw_j}\exp\left\{\dfrac{1}{2}{\bm{\mu}}^{\mathrm{T}}{M}^{-1}{\bm{\mu}}\right\}\int\exp \left\{- \dfrac{1}{2}\left(\mathbf x-{\bm{\mu}}\right)^{\mathrm T}{M}^{-1}\left(\mathbf x-{\bm{\mu}}\right)\right\} \mathrm{d}{\mathbf x}\\
&=C(\hat{t})e^{-\frac{1}{2}\sum_{j=1}^{q}t_{j}\vecw_j^\mathrm T K_j^{-1}\vecw_j}\exp\left\{\dfrac{1}{2}{\bm{\mu}}^{\mathrm{T}}{M}^{-1}{\bm{\mu}}\right\}(2\pi)^{n/2}(\det {M})^{1/2},
\end{aligned}
\end{equation}
}
\normalsize
where we have:
\begin{equation}
C(\hat{t})=\prod_{j=1}^{q} (2\pi)^{-nt_{j}/2} (\det K_j)^{-t_{j}/2}=(2\pi)^{-n{a}/2}\prod_{j=1}^{q}  (\det K_j)^{-t_{j}/2}.
\end{equation}
Note that we can alternatively write:
\begin{equation}
{\bm{\mu}}^{\mathrm{T}}{M}^{-1}{\bm{\mu}}=\left(\sum_{j=1}^{q}t_{j}K_j^{-1}\vecw_j\right)^{\mathrm{T}}M\left(\sum_{l=1}^{q}t_{l}K_l^{-1}\vecw_l\right).
\end{equation}
\section{Proof of Corollary \ref{fpowerintegralcorr}}
\label{fpowerintegralcorrapp}
\small
\begin{equation}
\begin{aligned}
  &\int  f_{\mathbf{X}}^a(\vecx) \mathrm{d}\mathbf{x}=\int  \left(\sum_{i=1}^{q}p_i g_i(\mathbf x)\right)^a \mathrm{d}\mathbf{x}
  =\int\sum_{{\substack{t_1+t_2+\dots+t_q=a\\t_1, t_2,\dots, t_q \geq 0}}}\binom{a}{t_1,t_2,\dots,t_q}\left(\prod_{i=1}^{q}p_i^{t_i}\right)\left[\prod_{j=1}^{q}g^{t_j}_j(\mathbf x)\right]\mathrm{d}\mathbf{x}\\
  &=\sum_{{\substack{t_1+t_2+\dots+t_q=a\\t_1, t_2,\dots, t_q \geq 0}}} \binom{a}{t_1,t_2,\dots,t_q}\left(\prod_{i=1}^{q}p_i^{t_i}\right)\int \prod_{j=1}^{q}g^{t_j}_j(\mathbf x)\mathrm{d}\mathbf{x}
=\sum_{{\substack{t_1+t_2+\dots+t_q=a\\t_1, t_2,\dots, t_q \geq 0}}} \binom{a}{t_1,t_2,\dots,t_q}\left(\prod_{i=1}^{q}p_i^{t_i}\right){D}(\hat t),
\end{aligned}
\end{equation}
\normalsize
where the last equality follows from Lemma \ref{productgausslemma}. Plugging into equation \eqref{entropyseries} gives the desired expression.

\section{Proof of Theorem \ref{taylortheo}}
\label{taylortheorapp}
\begin{equation}
\begin{aligned}
&f_{\vecX}Z^k=f_{\vecX}\sum_{b=0}^{k} \binom{k}{b}\left(-m^{-1}f_{\vecX}\right)^b=\sum_{b=0}^{k} \binom{k}{b}(-m)^{-b}f_{\vecX}^{b+1}=-m\sum_{b=0}^{k} \binom{k}{b}(-1)^{b+1}\dfrac{f_{\vecX}^{b+1}}{m^{b+1}}\\
&=-m\sum_{a=1}^{k+1} \binom{k}{a-1}(-1)^{a}\dfrac{f_{\vecX}^{a}}{m^{a}},
\end{aligned}
\end{equation}
and plugging into equation \eqref{flnftaylor}, we have:
\begin{equation}
\begin{aligned}
-f_{\vecX}\ln \dfrac{f_{\vecX}}{m}=\sum_{k=1}^{\infty}\dfrac{-m}{k}\sum_{a=1}^{k+1} \binom{k}{a-1}(-1)^{a}\dfrac{f_{\vecX}^{a}}{m^{a}}=\sum_{k=1}^{\infty}\sum_{a=1}^{k+1}\tilde{c}_{k,a}f_{\vecX}^{a},
\end{aligned}
\end{equation}
where $\tilde{c}_{k,a}$ is given by:
\begin{equation}
\tilde{c}_{k,a}=\binom{k}{a-1}\dfrac{(-1)^{a+1}}{km^{a-1}}.
\end{equation}
If we now take the first $C-1$ terms, we have:
\begin{equation}
\begin{aligned}
&-f_{\vecX}\ln \dfrac{f_{\vecX}}{m}=\sum_{k=1}^{C-1}\sum_{a=1}^{k+1}\tilde{c}_{k,a}f_{\vecX}^{a}+\cdots=\sum_{k=1}^{C-1}\sum_{a=1}^{C}\tilde{c}_{k,a}f_{\vecX}^{a}+\cdots\\
&=\sum_{a=1}^{C}\sum_{k=1}^{C-1}\tilde{c}_{k,a}f_{\vecX}^{a}+\cdots=\sum_{a=1}^{C}c_a^{\mathrm{Taylor}}f_{\vecX}^{a}+\cdots
\end{aligned}
\end{equation}
where we have:
\begin{equation}
c_{1}^{\mathrm{Taylor}}=\sum_{k=1}^{C-1}\tilde{c}_{k,1}=H_{C-1},
\end{equation}
where $H_k$ is the $k$-th Harmonic number. Furthermore, we have ($a \neq 1$):
\begin{equation}
c_{a}^{\mathrm{Taylor}}=\sum_{k=1}^{C-1}\tilde{c}_{k,a}=\sum_{k=1}^{C-1}\binom{k}{a-1}\dfrac{(-1)^{a+1}}{km^{a-1}}=\sum_{k=a-1}^{C-1}\binom{k}{a-1}\dfrac{(-1)^{a+1}}{km^{a-1}}=\dfrac{(-1)^{a+1}}{m^{a-1}}\dfrac{1}{a-1}\binom{C-1}{a-1},
\end{equation}
and the last equality follows from:
\begin{equation}
\sum_{k=a}^{C}\dfrac{1}{k}\binom{k}{a}=\sum_{\tilde{k}=0}^{C-a}\dfrac{1}{\tilde{k}+a}\binom{\tilde{k}+a}{a}=\sum_{\tilde{k}=0}^{C-a}\dfrac{1}{a}\binom{\tilde{k}+a-1}{a-1}=\dfrac{1}{a}\binom{C}{a},
\end{equation}
where we used the binomial identities:
\begin{align}
&\dfrac{1}{a}\binom{a}{b}=\dfrac{1}{b}\binom{a-1}{b-1},\quad a,b\neq 0,\\
&\sum_{j=0}^{m}\binom{n+j}{n}=\binom{n+m+1}{n+1}=\binom{n+m+1}{m}.
\end{align}
Plugging the coefficients $\{c_{a}^{\mathrm{Taylor}}\}$ into Corollary \ref{fpowerintegralcorr}, we have:
\begin{equation}
\begin{aligned}
&h(\vecX)\approx-\ln m+H_{C-1}\sum_{{\substack{t_1+t_2+\dots+t_q=1\\t_1, t_2,\dots, t_q \geq 0}}} \binom{1}{t_1,t_2,\dots,t_q}\left(\prod_{i=1}^{q}p_i^{t_i}\right){D}(\hat t)\\
&+\sum_{a=2}^{C}\dfrac{(-1)^{a+1}}{m^{a-1}}\dfrac{1}{a-1}\binom{C-1}{a-1}\sum_{{\substack{t_1+t_2+\dots+t_q=a\\t_1, t_2,\dots, t_q \geq 0}}} \binom{a}{t_1,t_2,\dots,t_q}\left(\prod_{i=1}^{q}p_i^{t_i}\right){D}(\hat t)\\
&=-\ln m+H_{C-1}+\sum_{a=1}^{C-1}\dfrac{(-1)^{a}}{m^{a}}\dfrac{1}{a}\binom{C-1}{a}\sum_{{\substack{t_1+t_2+\dots+t_q=a+1\\t_1, t_2,\dots, t_q \geq 0}}} \binom{a+1}{t_1,t_2,\dots,t_q}\left(\prod_{i=1}^{q}p_i^{t_i}\right){D}(\hat t)\\
&=-\ln m+H_{C-1}-m\sum_{a=1}^{C-1}\dfrac{B_{a+1}}{a}\binom{C-1}{a},
\end{aligned}
\end{equation}
where we used:
\begin{equation}
\sum_{{\substack{t_1+t_2+\dots+t_q=1\\t_1, t_2,\dots, t_q \geq 0}}} \binom{1}{t_1,t_2,\dots,t_q}\left(\prod_{i=1}^{q}p_i^{t_i}\right){D}(\hat t)=1.
\end{equation}

\section{Proof of Lemma \ref{generalpolyfitlemma}}
\label{generalpolyfitlemmaproof}
If we take the partial derivative of $E$ w.r.t $d_{\mathcal{I}, i}$, we get:
\begin{equation}
\dfrac{\partial E}{\partial d_{\mathcal{I}, i}}=\dfrac{-2}{b-a}\int_{a}^{b} w(s)(f(s)-\hat{f}(s))s^{i}\mathrm{d}s=\dfrac{-2}{b-a}\left(\int_{a}^{b} w(s)f(s)s^{i}\mathrm{d}s-\sum_{j=1}^{C}d_{\mathcal{I}, j}\int_{a}^{b} w(s)s^{i+j}\mathrm{d}s\right),
\end{equation}
and setting it to zero, we get:
\begin{equation}
\sum_{j=1}^{C}d_{\mathcal{I}, j}\int_{a}^{b} w(s)s^{i+j}\mathrm{d}s=\int_{a}^{b} w(s)f(s)s^{i}\mathrm{d}s,
\end{equation}
or written in matrix form as:
\begin{equation}
M_{\mathcal{I}}\vec{d}_{\mathcal{I}}=\vec{z}_{\mathcal{I}}.
\end{equation}
In order to show that our solution is a global minimum, we show that the Hessian matrix $H_{ij}=\dfrac{\partial E}{\partial d_{\mathcal{I}, i}\partial d_{\mathcal{I}, j}}$ is positive definite. Taking the partial derivatives, we have:
\begin{equation}
H_{ij}=\dfrac{2}{b-a}\int_{a}^{b}w(s)s^{i}s^j \mathrm{d}s,
\end{equation}
which is a Gram matrix (note that $w(s)>0$ and a Gram matrix is positive semi-definite). It is straightforward to see that $H$ is positive definite since $\{\sqrt{w(s)}s^i\}_{i=1}^{C}$ are linearly independent.

\section{Proof of Corollary \ref{Mtildecorr}}
\label{Mtildecorrproofapp}
From Lemma \ref{slogslemma} and equation \eqref{dmatrixeq}, we have:
\begin{equation}
\begin{aligned}
&\sum_{j}(M_{\mathcal{I}})_{ij}(\vec{d}_{\mathcal{I}})_j=(\vec{z}_{\mathcal{I}})_i \Rightarrow \sum_{j}\dfrac{b^{i+j+r+1}}{i+j+r+1}(\vec{d}_{\mathcal{I}})_j=\dfrac{b^{i+r+2}}{(i+r+2)^2}-\dfrac{b^{i+r+2}\ln b}{i+r+2}\\
&\Rightarrow \sum_{j}\dfrac{b^{i+j+r+1}}{i+j+r+1}(\vec{d}_{\mathcal{I}})_j+\dfrac{b^{i+r+2}\ln b}{i+r+2}=\dfrac{b^{i+r+2}}{(i+r+2)^2}\\
&\Rightarrow \sum_{j\neq 1}\dfrac{b^{i+j+r+1}}{i+j+r+1}(\vec{d}_{\mathcal{I}})_j+\dfrac{b^{i+r+2}}{i+r+2}(\vec{d}_{\mathcal{I}})_1+\dfrac{b^{i+r+2}\ln b}{i+r+2}=\dfrac{b^{i+r+2}}{(i+r+2)^2}\\
&\Rightarrow \sum_{j\neq 1}\dfrac{b^{j-1}}{i+j+r+1}(\vec{d}_{\mathcal{I}})_j+\dfrac{(\vec{d}_{\mathcal{I}})_1+\ln b}{i+r+2}=\dfrac{1}{(i+r+2)^2}\\
&\Rightarrow \sum_{j\neq 1}\dfrac{b^{j-1}}{i+j+r+1}(\vec{d}_{\mathcal{I}})_j+\left.\left(\dfrac{b^{j-1}(\vec{d}_{\mathcal{I}})_j+\ln b}{i+j+r+1}\right)\right|_{j=1}=\dfrac{1}{(i+r+2)^2}\\
&\Rightarrow \sum_{j}\tilde{M}_{ij}\tilde{d}_j=\tilde{z}_i.
\end{aligned}
\end{equation}


\section{Volume $V(s)$ for a single Gaussian}
\label{app:volume}

We consider the trivial `mixture' that consists of a single Gaussian centered on the origin,  with spherical
covariance matrix $\sigma^2 {\mathbf 1}$.
The volume $V(s)$ as defined in (\ref{entropyins}) is then given by
\begin{equation}
    V(s) = \int\rd \vecx\; \delta\left(s-{\cal N}_{0,\sigma^2 {\mathbf 1}}(\vecx)\right)
    \propto
    \int_0^\infty \rd r \; r^{n-1} \delta\left(s-\frac{\exp-\frac{r^2}{2\sigma^2}}{(\sigma\sqrt{2\pi})^n}\right).
\end{equation}
We apply the rule $\delta(s-y(r)) = \frac{ \delta(r-y^{\rm inv}(s)) }{|y'(r)|}$ and obtain
\begin{equation}
    V(s) \propto \int_0^\infty \rd r \; r^{n-1} \frac{  \delta\Big(r-  \sqrt{-2\sigma^2\ln[s(\sigma\sqrt{2\pi})^n]} \,\Big)  }
    { \frac{\exp-\frac{r^2}{2\sigma^2}}{(\sigma\sqrt{2\pi})^n} \cdot \frac{r}{\sigma^2}}
    \propto
    \frac1s  \left(\sqrt{\ln\frac1{s(\sigma\sqrt{2\pi})^n}} \right)^{n-2}
    = \frac1s  \Big(\ln\frac{f_{\rm max}}{s} \Big)^{\frac n2-1}.
\end{equation}
Here we have used that
$\frac{\exp-\frac{r^2}{2\sigma^2}}{(\sigma\sqrt{2\pi})^n} =s$
and that $ f_{\rm max}=\frac{1}{(\sigma\sqrt{2\pi})^n}$.
\end{document}